\documentclass[journal]{IEEEtran}
\usepackage{verbatim}
\usepackage{amsmath,amsfonts,amssymb}
\usepackage{color}
\usepackage[ruled]{algorithm2e}
\usepackage{placeins}
\usepackage{caption, subcaption}
\makeatletter
\def\ifundefined{\@ifundefined}
\makeatother

\usepackage{graphicx}

\usepackage{amsthm}
\theoremstyle{plain}
\newtheorem{theorem}{\textbf{Theorem}}
\newtheorem{lemma}[theorem]{\textbf{Lemma}}

\newtheorem*{conjecture*}{\textbf{Conjecture}}

\usepackage{stfloats}
\usepackage{blindtext}

\def\D{\boldsymbol{D}}

\def\F{\boldsymbol{F}}
\def\G{\boldsymbol{G}}

\def\I{\boldsymbol{I}}

\def\L{\boldsymbol{L}}

\def\R{\boldsymbol{R}}

\def\T{\boldsymbol{T}}
\def\U{\boldsymbol{U}}
\def\V{\boldsymbol{V}}

\def\f{\boldsymbol{f}}

\def\p{\boldsymbol{p}}
\def\q{\boldsymbol{q}}
\def\r{\boldsymbol{r}}

\def\u{\boldsymbol{u}}
\def\v{\boldsymbol{v}}

\def\bGamma{\boldsymbol{\Gamma}}

\def\bPhi{\boldsymbol{\Phi}}

\def\bphi{\boldsymbol{\phi}}

\def\0{\boldsymbol{0}}
\def\1{\boldsymbol{1}}

\begin{document}

\title{Performance Analysis of Automotive SAR With Radar Based Motion Estimation} 
\author{
    Oded Bialer and Tom Tirer
    \thanks{Oded Bialer is with General Motors Advanced Technical Center, Herzliya 46733, Israel (e-mail: oded.bialer8@gmail.com). 
    Tom Tirer is with the Center for Data Science, New York University, New York, USA. This work was done when he was with General Motors Advanced Technical Center, Herzliya 46733, Israel (e-mail: tirer.tom@gmail.com).}
}

\maketitle
\begin{abstract}
Automotive synthetic aperture radar (SAR) can achieve a significant angular resolution enhancement for detecting static objects, which is essential for automated driving. Obtaining high resolution SAR images requires precise ego vehicle velocity estimation. A small velocity estimation error can result in a focused SAR image with objects at offset angles. In this paper, we consider an automotive SAR system that produces SAR images of static objects based on ego vehicle velocity estimation from the radar return signal without the overhead in complexity and cost of using an auxiliary global navigation satellite system (GNSS) and inertial measurement unit (IMU). We derive a novel analytical approximation for the automotive SAR angle estimation error variance when the velocity is estimated by the radar. 
The developed analytical analysis closely predicts the true SAR angle estimation variance, and also provides insights on the effects of the radar parameters and the environment condition on the automotive SAR angle estimation error. We evaluate via the analytical analysis and simulation tests the radar settings and environment condition in which the automotive SAR attains a significant performance gain over the angular resolution of the short aperture physical antenna array. We show that, perhaps surprisingly, when the velocity is estimated by the radar the performance advantage of automotive SAR is realized only in limited conditions. Hence since its implementation comes with an increase in computation and system complexity as well as an increase in the detection delay it should be used carefully and selectively.
\end{abstract}

\section{Introduction}\label{INTRO}
Automotive radar is a key sensor for automated driving due to its large detection range and robustness to weather and light conditions. The radar's angular resolution is given by the wavelength divided by the antenna array aperture (diffraction limit) \cite{VAN_TREES}. State-of-the-art automotive MIMO radars operate at 77GHz and have an aperture of about 10cm thus they  reach an angular resolution of $1^{\circ}$ \cite{MIMO_1}-\cite{MIMO_3}, which is significantly lower than the angular resolution of LIDAR and camera and is insufficient for automated driving.

Various signal processing methods for angle of arrival estimation that achieve higher resolution than the diffraction limit have been proposed in literature. Sub-space methods such as ESPRIT \cite{ESPRIT_REF}, MVDR \cite{MVDR_REF}, and MUSIC \cite{MUSIC_REF_1}-\cite{MUSIC_REF_2}, have shown to obtain super high angular resolution when there are multiple non-coherent observations from the same target location. However, in the automotive scenario, the targets are moving from one observation to the other, and hence their angle is changing. Sparsity methods \cite{SPARSITY_OMP_REF}-\cite{SPARSITY_LASSO_REF}, and approximate maximum likelihood methods \cite{EM_REF}-\cite{IQML_REF}, have also shown better AOA performance than the diffraction limit, when the targets are point reflectors, the number of targets is known (in the case of maximum likelihood estimation), the signal to noise ratio (SNR) is high, and the array manifold is known precisely. In practice, these assumptions do not hold in the automotive scenario. The targets have extended shape with multiple reflection points, the SNR of objects at distance is low, and the array manifold is not precisely known due to temperature variations. 

Recently, deep neural network (DNN) methods have also been applied for angle estimation \cite{DNN_AOA}-\cite{DNN_AOA_2}. These approaches require a significant overhead of producing a large amount of annotated data for training the DNN, and may be non-robust for changes in the sensor and the environment. Another recent approach for increasing the angular resolution of automotive radar is by extending the physical antenna array aperture \cite{SHRDR_REF}. The wide aperture solution comes with an increase in sensor cost and is a challenge to mount on the vehicle.

A different approach for increasing the automotive radar angular resolution is the synthetic aperture radar (SAR) \cite{SAR_TUTORIAL}. In SAR a large synthetic aperture is generated by the motion of the vehicle. The received signals along the trajectory of the vehicle are coherently combined and generate a very high resolution radar image of static objects, which is essential for automated driving in order to accurately detect the position and shape of static obstacles such as parking vehicles, poles, curbs, etc. However, successfully generating a high resolution SAR image for the static objects in a scene requires knowing  precisely the positions of the antenna elements along the motion of the vehicle. 

Hu et al. \cite{SAR_ANAL_1} have Analyzed the automotive SAR resolution and showed that when the ego-velocity is known the SAR can be useful for detecting obstacles in the squint forward direction. Gishkori et al. \cite{SAR_ALG_1} have developed a modified back-projection (BP) algorithm and a compressed sensing based BP algorithm, which is based on known vehicle motion.

Various algorithms for automotive SAR that achieve high resolution images of static objects and rely on auxiliary sensors (other than the radar) to obtain precise ego-motion of the vehicle have been proposed in the literature.
In \cite{SAR_BP_1} the BP algorithm was applied with ego vehicle velocity obtained from accurate odometry. The modified fast factorized BP algorithm with the use of Global Navigation Satellite System (GNSS), and Inertial Measurement Units (IMU) has been proposed in \cite{SAR_BP_2}. The IMU measures the vehicle force, and angular rate with  accelerometers and gyroscopes. The range migration algorithm (RMA) with velocity obtained from vehicle sensors has also been proposed \cite{SAR_RMA_1},\cite{SAR_RMA_2}.
In \cite{SAR_EFF_IMP_1} an efficient SAR implementation was introduces with GNSS and IMU sensors. A SAR operating at 300GHz with very high bandwidth of 40GHZ with velocity that is obtained from a GNSS and IMU unit showed very high resolution performance \cite{SAR_EFF_IMP_2}. 
A technique for fusing information from GNSS, IMU, odometers and steering angle sensors to support the automotive SAR system was introduces in \cite{SAR_EFF_IMP_3}.

Relying on auxiliary sensors to obtain the ego vehicle motion increases system cost, has a complexity  overhead of sensors synchronizations, and often is not robust when there is lack of GNSS receptions (e.g. parking lots or tunnels). Therefore, several SAR algorithms have been proposed that are based on the radar's self ego-velocity estimation. 
The ego vehicle velocity can be estimated from the Doppler and angle measurements that are  obtained from the radar's physical antenna array \cite{VEL_EST_1}-\cite{VEL_EST_3}.

Manzoni et al. \cite{SAR_RADAR_VEL_EST_1} have shown that automotive SAR with motion compensation by velocity estimation from the radar can reduce the residual error of the GNNS and IMU. The work in \cite{SAR_RADAR_VEL_EST_2} demonstrated that enhanced angular resolution is achievable with SAR BP algorithm that is based on velocity estimation with radar only.
Gao et al. \cite{SAR_TVT_REF} introduce an efficient implementation of the SAR BP algorithm for automotive radar, which is based on ego vehicle motion estimation from the radar returns. 
Gisder et al. \cite{SAR_WITH_RADAR_VEL_EST_3} have shown that when the velocity is estimated by the automotive radar then the residual velocity estimation errors result in angle offset of the objects in the SAR image. The SAR angular offset due to velocity error has been analyzed for airborne radars \cite{SAR_ERR_ANALYSIS_1}-\cite{SAR_ERR_ANALYSIS_4}. When the velocity estimation error is large then the objects in the SAR image are blurred (unfocused) with relatively low intensity, and also have an angle offset. On the other hand, when the velocity error is small the SAR image is focused with high intensity, but the objects have an angle offset.  
The angular error analysis in \cite{SAR_ERR_ANALYSIS_1}-\cite{SAR_ERR_ANALYSIS_4} considers the velocity error as a parameter that is typically given from the accuracy of the GNNS and IMU systems. Furthermore, the airborne platform velocity and the velocity estimation are assumed constant over the SAR coherent integration interval, and the geometry is of an airborne scenario (not an automotive scenario).

In this paper, we consider automotive SAR for imaging of static objects, with ego-velocity that is estimated only from the radar. We address the counteractive interplay between the angle and the radar's velocity estimation errors, which results in an angle estimation error of the automotive SAR. We derive a novel  analytical analysis for the automotive SAR angle estimation error when the velocity is estimated by the radar and varies during the coherent integration interval of the SAR. The main contribution of the paper is an analytical tool that is useful for accurately predict the SAR angle estimation error as a function of the automotive radar parameters, vehicle speed, and the reflections in the scenario. We make use of the developed analytical tool to evaluate the system parameters and the environment condition in which the automotive SAR attains a significant performance gain over the short aperture physical antenna array. Somewhat surprising, we show that when the velocity is estimate by the radar the performance advantage of automotive SAR is realized in rather limited conditions, and since its implementation comes with an increase in computation complexity, and increases the detection delay (due to longer integration interval), it should be used carefully and selectively.

The remaining of the paper is organized as follows. Section \ref{SYS_MODEL_SEC} describes the system model and provides background on the radar's velocity estimation and the automotive SAR processing. Section \ref{VEL_AZ_AMB_SEC} explains the SAR velocity-angle ambiguity issue. In Section \ref{PERFORMANCE_ANALY_SEC} we derive an analytical analysis for the automotive SAR performance. Section \ref{RES_SEC} corroborates our formulas with numerical experiments, and Section \ref{CONC_SEC} gives the paper conclusions. As for notations that are used in this paper, we use the tilde sign above a parameter to distinguish between its true value and its hypothesis value: for example $\theta$ is the true value of a parameter and $\tilde{\theta}$ is its hypothesis value. 

\section{System Model and Automotive SAR Background}\label{SYS_MODEL_SEC}
The automotive radar system that is considered in this paper is depicted at Fig.~\ref{SYS_MODEL}. The radar has a relative short physical antenna array with multiple transmit and receive antennas (MIMO radar) \cite{MIMO_2}. The transmitters transmit a sequence of FMCW fast chirps \cite{FAST_FMCW_1}-\cite{FAST_FMCW_2} at 77GHz carrier frequency. The received signal is down-converted by multiplying it with the transmitted sequence of chirps. 
The down-converted signal is partitioned into frames with short duration of about $20ms$. It is assumed that the frames are sufficiently short such that the velocity can be approximated as constant within each frame. Standard matched filter processing in range, Doppler and angle is applied per each frame and the output is a data cube per frame with dimensions of range, angle and Doppler. The matched filter processing includes FFT for each down-converted chirp (range stretch processing), FFT along chirps per each range bin (Doppler stretch processing), and beamforming per each range-Doppler bin \cite{MF_PROCESS_1}-\cite{MF_PROCESS_2}. 

When the vehicle is moving, the position of the radar changes along the frames and thus by coherently combining a sequence of frames (e.g. over $100ms-200ms$) an extended radar aperture (synthetic aperture) is created with potentially enhanced angular resolution.    
In this paper, we consider SAR processing only for reflections from static objects because the relative velocity estimation of static objects is significantly more accurate than for dynamic objects. The reason is that all the reflections from static objects have the same velocity relative to the radar, which is the negative ego vehicle velocity. Hence typically static objects provide a relatively larger number of reflection points with a large angular spread (compared to individual dynamic objects) that can be used for estimating their common velocity. 

We consider the radar ego vehicle velocity as randomly varying along a sequence of short frames (but constant within each frame). It is assumed that the ego-velocity is estimated in each radar frame only from the radar detected reflections from static objects, which are obtained from the matched filter output of each frame. 
Other auxiliary sensors such as accelerometer, gyroscope, GNNS, are not used for the velocity estimation. 
The radar detections in each frame are obtained by finding the high intensity peaks in the matched filter output, which are above the noise level threshold. The detections of static objects can be classified as the detections that are associated to the largest cluster of Doppler frequencies. 

A common approach for estimating the ego vehicle velocity in the $n^{th}$ frame, ${\v}_n$, from the radar detections of static objects, which has been used in recent works \cite{VEL_EST_1}-\cite{VEL_EST_3}, is to minimize the following least squares cost function w.r.t. the velocity hypothesis $\tilde{\v}_n$ 
\begin{align}\label{VEL_LS}
    Q(\tilde{\v}_n; \f_n, {\bphi}_n) = \frac{1}{2} \| \G(\bphi_n) \tilde{\v}_n - \f_n \|_2^2,
\end{align}
where $\f_n=[f_{n,1}, \ldots, f_{n,K}]^T$, $\bphi_n=[\phi_{n,1}, \ldots, \phi_{n,K}]^T$, are the Dopplers and angles, respectively, of the $K$ detected reflection points from static objects in the $n^{th}$ frame, and
\begin{align}
\label{Eq_G_def}
\G(\bphi_n) \triangleq \frac{2}{\lambda} \begin{bmatrix}
           \p^T(\phi_{n,1}) \\
           \vdots \\
           \p^T(\phi_{n,K})
         \end{bmatrix} =
         \frac{2}{\lambda}\begin{bmatrix}
           \sin{\phi_{n,1}} & \cos{\phi_{n,1}} \\
           \vdots & \vdots \\
           \sin{\phi_{n,K}} & \cos{\phi_{n,K}}
         \end{bmatrix}, 
\end{align}
where $\lambda$ is the wavelength and
$\p(\phi) \triangleq \begin{bmatrix}
           \sin{{\phi}} &
           \cos{{\phi}}
         \end{bmatrix}^T$ is a unit vector pointing in the direction from the boresight of the radar to a reflection point at angle $\phi$. 
Note that the angle measurements in \eqref{VEL_LS}, $\boldsymbol{\bphi}_n$, are obtained from the automotive radar's physical antenna array that has a short aperture and not from the extended synthetic aperture array that is obtained from the motion of the vehicle. Therefore, these angle measurements are relative inaccurate and, as we show in this paper, result in a velocity estimation error that degrades the SAR angle estimation performance.

The SAR high angular resolution output is obtained by coherently combining the matched filter outputs of $N$ consecutive frames for a given hypothesis of reflection point range and angle at the beginning of the integration interval, and a given velocity estimation for each frame. The reflection point range migration along the frames is determined from the velocity estimation of each frame. Given the velocity estimations and the initial range and angle hypothesis the reflection point range, Doppler and angle bin of each frame is computed and the matched filer outputs of these bins are coherently combined by compensating for the phase offsets due to the range migration along the frames. For more details on the SAR processing please refer to \cite{SAR_TVT_REF}. 

Next, we will formulate the SAR processing described above. Denote the reflection point true range and angle at the $n^{th}$ frame as $\gamma_n$ and $\theta_n$, respectively, and the hypothesis of reflection point range and angle at frame $n$ as $\tilde{\gamma}_n$ and $\tilde{\theta}_n$, respectively. Denote the true reflection point velocity vector and estimated velocity vector at the $n^{th}$ frame as $\v_n$ and $\tilde{\v}_n$, respectively.
Let $\V \triangleq \begin{bmatrix}
           \v_{0}^T \\
           \vdots \\
           \v_{N-1}^T
         \end{bmatrix} \in \mathbb{R}^{N \times 2}$ and $\tilde{\V} \triangleq \begin{bmatrix}
           \tilde{\v}_{0}^T \\
           \vdots \\
           \tilde{\v}_{N-1}^T
         \end{bmatrix}$ be matrices that encapsulate the ground truth and estimated velocities in all of the frames, respectively.
To simplify the notation throughout the paper we denote the true reflection point range and angle at the beginning of the integration interval as $\gamma$, and $\theta$, respectively, and the hypothesis of the reflection point range and angle at the beginning of the integration interval as $\tilde{\gamma}$ and $\tilde{\theta}$, respectively, i.e., $\theta \triangleq \theta_0$, $\gamma \triangleq \gamma_0$, $\tilde{\theta} \triangleq \tilde{\theta}_0$, and $\tilde{\gamma} \triangleq \gamma_0$. The SAR coherent combining of the matched filter outputs of $N$ frames for the hypothesis $\tilde{\gamma},\tilde{\theta}$, and the velocity estimations $\tilde{v}_0,..,\tilde{v}_{N-1}$ is given by
\begin{equation} \label{SAR_OUT_5}
\mu(\tilde{\gamma},\tilde{\theta},\tilde{\V})
\propto \\
\left|\sum_{n=0}^{N-1}y_n(\tilde{\gamma}_n,\tilde{\theta}_n,\tilde{\v}_n)e^{\frac{-j4\pi}{\lambda}r_n(\tilde{\V},\tilde{\theta})}\right|,
\end{equation}
where $y_n(\tilde{\gamma}_n,\tilde{\theta}_n,\tilde{\v}_n)$ is the $n^{th}$ frame matched filter output at range $\tilde{\gamma}_n$, angle $\tilde{\theta}_n$, and velocity $\tilde{\v}_n$ (which determines the Doppler frequency), and
\begin{equation} \label{RANGE_FORMULA}
r_n(\tilde{\V},\tilde{\theta})=T_f\sum_{k=0}^{n}\boldsymbol{p}^T(\tilde{\theta})\tilde{\v}_k,
\end{equation}
is the reflection point range migration hypothesis at the $n^{th}$ frame with respect to the beginning of the integration interval, where $T_f$ is the frame duration, and $\boldsymbol{p}(\theta)=\left[ \begin{array}{cc} \sin(\theta) & \cos(\theta) \end{array} \right]^T$ is a unit vector pointing in the direction from the boresight of the radar in the beginning of the integration interval to the reflection point at angle $\theta$. The SAR processing in \eqref{SAR_OUT_5} is described in Fig.~\ref{SAR_OUT_5}. 

The SAR output image (in range-angle domain) is generated by estimating the velocity in each frame and then calculating \eqref{SAR_OUT_5} for different hypothesis of $\tilde{\gamma}$, and $\tilde{\theta}$ given the velocity estimations. The peaks of the SAR output image that are above a noise level threshold are detected as reflection points \cite{CFAR_REF_1}-\cite{CFAR_REF_3}. Note the difference between two angle notations: $\theta$ and $\phi_{n,i}$. We use the notation $\theta$ for a reflection point angle that is estimated by the SAR and $\tilde{\theta}$ for the SAR hypothesis angle, while the notation $\phi_{n,i}$ refers to the angle of a detection that is obtained from the radar physical antenna array (see Fig.\ref{SYS_MODEL}) and used in \eqref{VEL_LS} for the velocity estimation.  

Fig.~\ref{SAR_KNOWN_VEL} shows an example of a SAR output image (with intensity in dB scale) of a simulated target vehicle at $28m$ and angle of $6^{\circ}$. The vehicle is illustrated in the image with white dotes. The black asterisks in the image are the simulated reflection points from the vehicle. The target vehicle was static, and the radar host vehicle was driving straight at zero angle with speed of $10m/s$. The velocity was perfectly known, and the SAR integration was over 15 frames of $20ms$ (i.e. the total integration time was $300ms$).

For reference we present in Fig.~\ref{SINGLE_FRMAE_MF} the reflection intensity image that is obtained for the same vehicle as in Fig.~\ref{SAR_KNOWN_VEL} but this time with conventional matched filter processing of a single frame with duration of $20ms$ (the image shows the output matched filter in range-azimuth with maximum over Doppler domain). In Fig.~\ref{SINGLE_FRMAE_MF} the angular resolution was $10^{\circ}$, which was obtained from the physical antenna array aperture of $2.2cm$. In Fig.~\ref{SAR_KNOWN_VEL} the angular resolution was obtained from a synthetically extended array of $3m$. By comparing  Fig.~\ref{SAR_KNOWN_VEL} with Fig.~\ref{SINGLE_FRMAE_MF} it is realized that the SAR achieves a significant improvement in the angular resolution compared to the physical antenna array resolution. However, recall that in Fig.~\ref{SAR_KNOWN_VEL} the velocity was perfectly known. When the velocity is unknown, the error in the velocity estimation causes a significant degradation in the SAR angular accuracy, 
as will be discussed and analyzed in the following sections.

\begin{figure}
  \centering
  \begin{subfigure}[b]{1\linewidth}
    \centering\includegraphics[width=280pt]{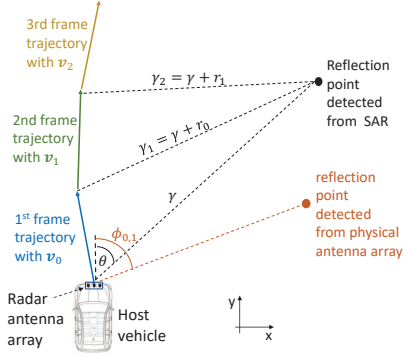}  
  \end{subfigure}
  \caption{Illustration of a synthetic aperture that is generating from three frames with velocities $\v_1, \v_2$ and $\v_3$. The SAR coherent integration is performed for the reflection point at angle $\theta$ and range $\gamma$ with respect to the beginning of the integration interval (black point). The velocity in each frame is estimated from the reflection points that are detected by the physical antenna array (e.g. brown point at angle $\phi_1$).}
\label{SYS_MODEL}
\end{figure}
\begin{figure}
  \centering
  \begin{subfigure}[b]{1\linewidth}
    \centering\includegraphics[width=280pt]{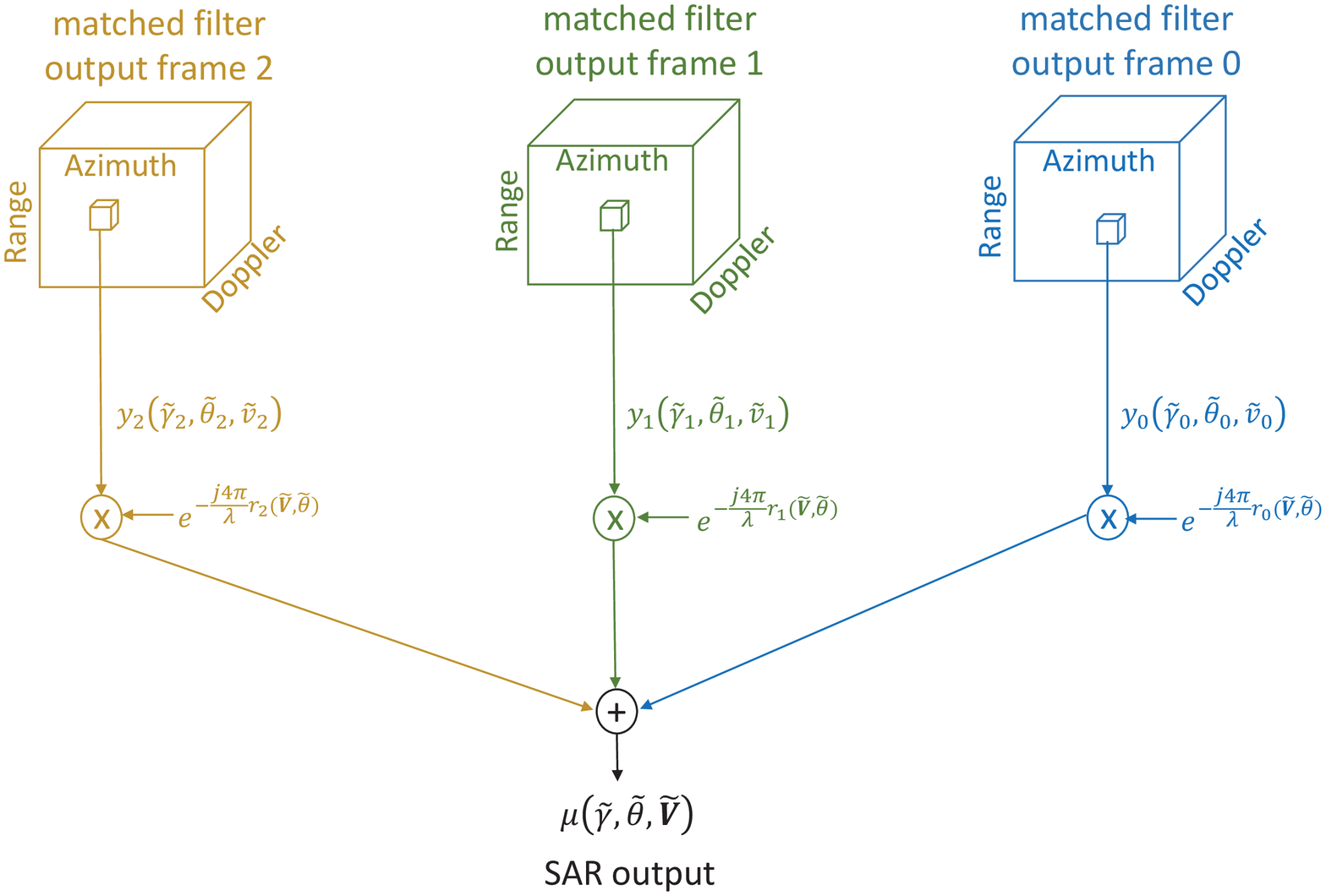}    
  \end{subfigure}
  \caption{Diagram of SAR coherent integration for three frames. A matched filter range-Doppler-angle data cube for each frame is computed. The matched filer outputs corresponding to the target's range, angle and velocity hypothesis of each frame are extracted and coherently combined by compensating for the phase shift along frames that is due to the target range migration hypothesis $r_n(\tilde{\V},\tilde{\theta})$.}
\label{SAR_PROCESSING_DIAGRAM}
\end{figure}

\begin{figure}
  \centering
  \begin{subfigure}[b]{1\linewidth}
    \centering\includegraphics[width=260pt]{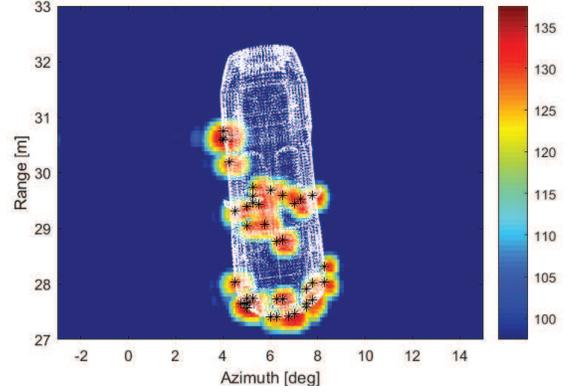}    
  \end{subfigure}
  \caption{Image of SAR output intensity in dB for integration duration of $300ms$ and known velocity of $10m/s$. The white dots illustrate the vehicle target and the black asterisk are the simulated reflection points from the vehicle.}
\label{SAR_KNOWN_VEL}
\end{figure}
\begin{figure}
  \centering
  \begin{subfigure}[b]{1\linewidth}
    \centering\includegraphics[width=260pt]{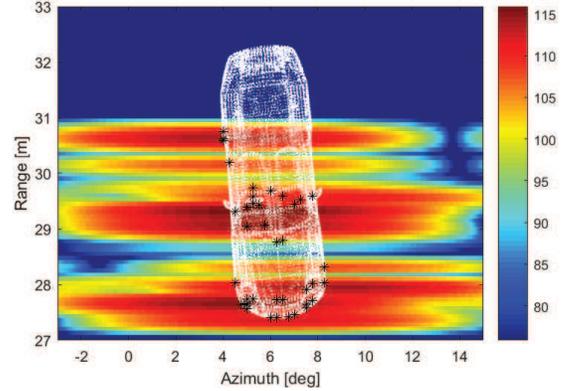}   
  \end{subfigure}
  \caption{Image of matched filter output intensity in dB (with maximum over Doppler) of a single frame with duration of $20ms$ and the same simulated vehicle as in Fig.~\ref{SAR_KNOWN_VEL}.}
\label{SINGLE_FRMAE_MF}
\end{figure}

\section{SAR Velocity-Angle Ambiguity}\label{VEL_AZ_AMB_SEC}
A fundamental issue in SAR is that a small error in the velocity estimation results in high intensity in the SAR output at a wrong reflection angle. 
Fig.~\ref{SAR_UNKNOWN_VEL} shows the SAR output image for the same case as in Fig.~\ref{SAR_KNOWN_VEL}, except that there was a very small velocity error of $0.03m/s$ (in Fig.~\ref{SAR_KNOWN_VEL} there was no velocity error). The images in Fig.~\ref{SAR_KNOWN_VEL} and Fig.~\ref{SAR_UNKNOWN_VEL} both seem focused with high intensity. However, the high reflection intensity peaks in Fig.~\ref{SAR_KNOWN_VEL} are at the angles of the correct reflection points (the black asterisks), while the reflection peaks in Fig.~\ref{SAR_UNKNOWN_VEL} have an angular offset due to the velocity error. Furthermore, it is realized that it is difficult to identify in which one of the images the reflection points are at the correct angles without knowing the ground truth. Therefore, auto-focusing algorithms \cite{SAR_AUTOFOCUS_1}-\cite{SAR_AUTOFOCUS_3} will be unsuccessful in this case since both images are focused.

Let us understand the SAR angular error phenomena from the SAR coherent combining formula in \eqref{SAR_OUT_5}. When the hypotheses of the reflection point range and angle at the beginning of the integration interval $\tilde{\gamma}, \tilde{\theta}$ have a large offset from the true range and angle (large offset with respect to the range and angle resolution of the individual
frame) and/or the velocity estimation error is large, then the intensity of the matched filter outputs, 
$y_n(\tilde{\gamma}_n,\tilde{\theta}_n,\tilde{\v}_n)$, is significantly lower than the matched filter outputs at the true reflection point parameters, $y_n(\gamma_n,\theta_n,\v_n)$. As a result, the SAR output in \eqref{SAR_OUT_5} has relatively low intensity. In this case, the SAR image can be distinguished as unfocused (smeared with low intensity) and the velocity estimation could be corrected until the image is focused (but the objects are not necessarily at the correct angle as explained next). On the other hand, when the reflection point range and angle hypotheses, $\tilde{\gamma}, \tilde{\theta}$ have a small offset from the true reflection point range and angle,$\gamma, \theta$, and also the velocity estimation error is small then the frames' matched filter outputs of the hypotheses  $\tilde{\gamma},\tilde{\theta},\tilde{\v}_n$ ($n=0,1,..,N-1$) are approximately equal to the matched filter outputs of the true reflection point
parameters, $\gamma,\theta,\v_n$, (because the errors in the reflection point range, angle and velocity are insignificant with respect to the individual frame resolution), i.e.,  
\begin{equation} \label{MF_APPROX}
y_n(\tilde{\gamma}_n,\tilde{\theta}_n,\tilde{\v}_n)\approx y_n(\gamma,\theta,\v_n).
\end{equation}
We can express the match filter output for the true reflection point parameters at frame $n>0$ as a function of the matched filter output at the first, $n=0$, as follows
\begin{equation} \label{MF_FUNC_OF_1ST_FRAME}
y_n(\gamma_n,\theta_n,\v_n) = y_0(\gamma,\theta,\v_0)e^{\frac{j4\pi}{\lambda}r_n(\V,\theta)},
\end{equation}
where 
\begin{equation} \label{RANGE_TRUE}
r_n(\V,\theta)=T_f\sum_{k=0}^{n}\boldsymbol{p}^T(\theta)\v_k,
\end{equation}
is the reflection point range migration at the $n^{th}$ frame with respect to the first frame. 
Substituting \eqref{MF_FUNC_OF_1ST_FRAME} into \eqref{MF_APPROX} we have that 
\begin{equation} \label{SMALL_VEL}
y_n(\tilde{\gamma}_n,\tilde{\theta}_n,\tilde{\v}_n)\approx y_0(\gamma,\theta,\v_0)e^{\frac{j4\pi}{\lambda}r_n(\V,\theta)},
\end{equation}
and by substituting \eqref{SMALL_VEL} into \eqref{SAR_OUT_5} we obtain that 
\begin{multline} \label{SAR_OUT_2}
\mu(\tilde{\gamma},\tilde{\theta},\tilde{\V})
\propto \\
\left|y_0(\gamma,\theta,\v_0)\sum_{n=0}^{N-1}e^{\frac{-j4\pi}{\lambda}\left(r_n(\tilde{\V},\tilde{\theta})-r_n(\V,\theta)\right)}\right| \propto \\
\left|\sum_{n=0}^{N-1}e^{\frac{-j4\pi}{\lambda}\left(r_n(\tilde{\V},\tilde{\theta})-r_n(\V,\theta)\right)}\right|.
\end{multline}
From \eqref{SAR_OUT_2} it is realized that the SAR output has high intensity when the difference between the reflection point true range migration and hypothesis range migration, $r_n(\tilde{\V},\tilde{\theta})-r_n(\V,\theta)$, is small. This difference is equal to zero when the reflection point angle hypothesis and velocity estimation are equal to the ground truth, i.e., $\tilde{\theta}=\theta$, and $\tilde{\v}_n=\v_n$. However, the range difference may also be close to zero when there is an error in the velocity estimation, i.e., $\tilde{\v}_n \ne \v_n$. From \eqref{RANGE_FORMULA} it can be realized that when $\tilde{\theta}=\theta$ then an error in the velocity causes a range offset. However, this offset can be eliminated by introducing an offset in $\tilde{\theta}$ with respect to $\theta$. Thus, when there is a small velocity error then \eqref{SAR_OUT_2} has high intensity at an angle offset from the true reflection point angle. 

Figs.~\ref{SAR_KNOWN_VEL}-\ref{SAR_UNKNOWN_VEL} show that the SAR can theoretically achieve very high angular resolution and accuracy when the velocity is known. However, when the velocity is estimated by the radar then even a small velocity estimation error can  result in a significant SAR angular error. In the next section we derive an analytical analysis for the SAR angular estimation error when the ego vehicle velocity is estimated by the radar. 

\begin{figure}
  \centering
  \begin{subfigure}[b]{1\linewidth}
    \centering\includegraphics[width=260pt]{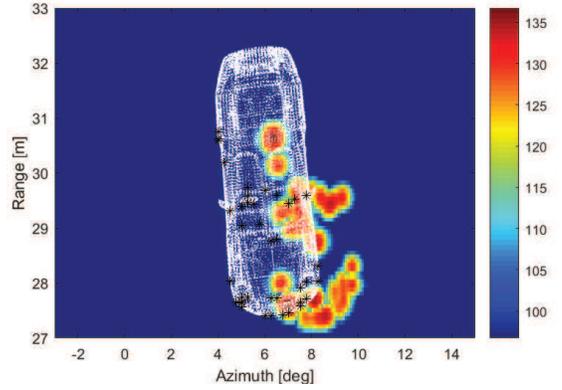}    
  \end{subfigure}
  \caption{Image of SAR output intensity in dB for the same simulation as in Fig.~\ref{SAR_KNOWN_VEL} except that there is a small velocity error of $0.03m/s$.}
\label{SAR_UNKNOWN_VEL}
\end{figure}

\section{Analytical Analysis of SAR Angular Error}\label{PERFORMANCE_ANALY_SEC}

In this section we derive an analytical analysis for the SAR angular estimation performance when the velocity is estimated from the radar detections (that are obtained from the physical antenna array). 
This 
section is organized as follows. In Section~\ref{sec:analysis1}, we analytically link the angular estimation error of the SAR to the velocity estimation error via Taylor expansion under an ambiguity assumption. We then translate the obtained formula to a relation between the angle error covariance and the velocity error covariance. In Section~\ref{sec:analysis2} we derive an analytical expression for the velocity error covariance. 
In Section~\ref{ANALYSIS_IMPLICATIONS} we formulate a full analytical expression for the SAR angle estimation error (appears in \eqref{SAR_ANGLE_ERR_FULL}) by integrating the results of the two preceding subsections.
Furthermore, in Section~\ref{ANALYSIS_IMPLICATIONS} we reach a 
significantly simplified SAR angular error formula (appears in \eqref{Eq_asymp_approx3}) under two additional assumptions on the setting.
From this formula we gain insights on the effects of the system parameters on 
the performance.

\subsection{Linking Angle and Velocity Error Covariances}
\label{sec:analysis1}

Let us start by highlighting the key idea of the analysis. 
Inspecting \eqref{SAR_OUT_2}, we see that 
\begin{multline}
    \left|\sum_{n=0}^{N-1}e^{\frac{-j4\pi}{\lambda}(r_n(\tilde{\V},\tilde{\theta})-r_n(\V,\theta))}\right| = \\
    \left| \sum_{n=0}^{N-1}e^{\frac{-j4\pi}{\lambda}(r_n(\tilde{\V},\tilde{\theta})-r_n(\V,\theta) - \alpha)}\right|
\end{multline}
for any $\alpha \in \mathbb{R}$.
This implies that an ambiguity occurs when the range measurements associated with the true velocities and angle, $\r(\V,\theta)=[r_0(\V,\theta),..,r_{N-1}(\V,\theta)]^T$, are (nearly) the same as those that match a different pair, $\r(\tilde{\V},\tilde{\theta})=[r_0(\tilde{\V},\tilde{\theta}),..,r_{N-1}(\tilde{\V},\tilde{\theta})]^T$, 
up to having a different mean (i.e., a constant phase offset, $\alpha$). 
Formally, an ambiguity is resulted from
\begin{align}
\label{Eq_ambiguity_cond}
    \bPhi_N \r(\tilde{\V},\tilde{\theta}) - \bPhi_N \r(\V,\theta) \approx \0,
\end{align}
where $\bPhi_N \in \mathbb{R}^{N \times N}$ is a projection matrix (of rank $N-1$) that subtracts the mean from any $N \times 1$ vector on which it is applied
$$
\bPhi_N \triangleq \I_N - \frac{1}{N}\1_N\1_N^T.
$$
Recall from Section \ref{VEL_AZ_AMB_SEC} that the SAR angle estimation error is caused when there is a small range offset with respect to the ground truth range that is due to a small velocity error. This motivates us to base our analysis on the first order Taylor expansion of $\r(\tilde{\V},\tilde{\theta})$ around the true $\r(\V,\theta)$. 

The first order Taylor expansion of $\r(\tilde{\V},\tilde{\theta})$ around the true $\r(\V,\theta)$ is given by
\begin{multline}
\label{Eq_analysis_taylor}
    \r(\tilde{\V},\tilde{\theta}) \approx \\
    \r(\V,\theta) + 
             \begin{bmatrix}
           \mathrm{Tr}\left ( \frac{\partial}{\partial \V^T} r_0(\V,\theta) \delta{\V} \right ) \\
           \vdots \\
           \mathrm{Tr}\left ( \frac{\partial}{\partial \V^T} r_{N-1}(\V,\theta) \delta{\V} \right )
         \end{bmatrix}
    + \frac{\partial}{\partial \theta} \r(\V,\theta) \delta \theta,
\end{multline}
where $\delta{\V} = \tilde{\V} - \V$ and $\delta \theta = \tilde{\theta} - \theta$ denote the velocity and angle errors, respectively.
Substituting \eqref{Eq_analysis_taylor} in \eqref{Eq_ambiguity_cond} yields
\begin{align}
\label{Eq_analysis_taylor2}
    \0 \approx  
             \bPhi_N \begin{bmatrix}
           \mathrm{Tr}\left ( \frac{\partial}{\partial \V^T} r_0(\V,\theta) \delta{\V} \right ) \\
           \vdots \\
           \mathrm{Tr}\left ( \frac{\partial}{\partial \V^T} r_{N-1}(\V,\theta) \delta{\V} \right )
         \end{bmatrix}
    + \bPhi_N \frac{\partial}{\partial \theta} \r(\V,\theta) \delta \theta.
\end{align}
Equation \eqref{Eq_analysis_taylor2} further displays the interplay between the errors in case of ambiguity: an error component in a certain range direction which is due to inaccurate velocity is ``compensated"/``masked" by an error component in the opposite range direction which is due to inaccurate angle.

We turn now to compute the derivatives that appear in \eqref{Eq_analysis_taylor2}.
First, based on \eqref{RANGE_FORMULA}, let us represent the vector of $N$ ground truth range offsets as 
\begin{equation} \label{RANGE_GT_VEC}
\boldsymbol{r}(\boldsymbol{V},\theta) \triangleq \left[\begin{array}{c} r_0(\V,\theta) \\ : \\ r_{N-1}(\V,\theta) \end{array}\right]=T_f\L_N \V\p(\theta),
\end{equation}
where $\L_N$ is a $N \times N$ lower triangular matrix of ones, e.g., $\L_3 = \begin{bmatrix} 1&0&0 \\ 1&1&0 \\ 1&1&1 \end{bmatrix} $. 
With these notations, we compute the derivatives as follows
\begin{align}
\label{Eq_analysis_der_az}
    \frac{\partial}{\partial \theta} \r(\V,\theta) &= T_f \L_N \V \p'(\theta)
    \\
\label{Eq_analysis_der_vel}
    \frac{\partial}{\partial \V^T} r_n(\V,\theta) &= T_f \frac{\partial}{\partial \V^T} \left ( \L_N[n,:] \V           
         \p(\theta)
         \right ) \\ \nonumber
         &= T_f 
         \p(\theta) \L_N[n,:]
\end{align}
where $\p'(\theta)=\begin{bmatrix}
           \cos{{\theta}} \\
           -\sin{{\theta}}
         \end{bmatrix}$
and $\L_N[n,:] \in \mathbb{R}^{1 \times N}$ denotes the $n$-th row of $\L_N$. 
By substituting \eqref{Eq_analysis_der_az} and \eqref{Eq_analysis_der_vel} in \eqref{Eq_analysis_taylor2}, and making use of the identity 
\begin{multline}
\label{Eq_analysis_taylor_vel}
             \begin{bmatrix}
           \mathrm{Tr}\left ( 
         \p(\theta) \L_N[1,:] \delta{\V} \right ) \\
           \vdots \\
           \mathrm{Tr}\left ( 
         \p(\theta) \L_N[N,:] \delta{\V} \right )
         \end{bmatrix}= \\
             \begin{bmatrix}
           \L_N[1,:] \delta{\V} 
         \p(\theta) \\
           \vdots \\
           \L_N[N,:] \delta{\V} 
         \p(\theta) 
         \end{bmatrix} 
= \L_N \delta{\V} \p(\theta),
\end{multline}
we arrive at
\begin{align}
\label{Eq_analysis_taylor3}
     \bPhi_N \L_N \V
         \p'(\theta) \delta \theta
&\approx
- \bPhi_N \L_N \delta{\V} \p(\theta)
         \\ \nonumber
&= - \bPhi_N \L_N \left (  
         \p(\theta)^T \otimes \I_N \right ) \delta \v_{1:N} 
\end{align}
where $\delta \v_{1:N} \triangleq \mathrm{vec}(\delta{\V}) \in \mathbb{R}^{2N \time 1}$ and
$\otimes$ denotes the Kronecker product.
Defining the $K \times 1$ vector
\begin{align}
    \label{Eq_analysis_u_vec}
    \u \triangleq \bPhi_N \L_N \V \p'(\theta),
\end{align}
we can isolate $\delta \theta$ by multiplying both sides of \eqref{Eq_analysis_taylor3} from left by $\u^T/\|\u\|_2^2$. This yields
\begin{align}
    \label{Eq_analysis_az_vel}
\delta \theta \approx - \frac{1}{\|\u\|_2^2} \u^T \L_N \left (  
         \p(\theta)^T \otimes \I_N \right ) \delta \v_{1:N},
\end{align}
where we used the fact that $\bPhi_N^T \bPhi_N = \bPhi_N \bPhi_N = \bPhi_N$, and thus $\u^T \bPhi_N \L_N = \u^T \L_N$.

Finally, we obtain (under zero mean error assumption) our key formula for the relation between the variance of the angle error, $\mathrm{var}(\delta \theta)$, and the velocity error covariance matrix, $\mathrm{Cov}(\delta \v_{1:N})$, which is given by 
\begin{align}
    \label{Eq_analysis_az_vel_var}
\mathrm{var}(\delta \theta) &\approx \frac{1}{\|\u\|_2^4} \u^T \L_N \left (  
         \p(\theta)^T \otimes \I_N \right ) \times \\ \nonumber
&\hspace{25mm}         
         \mathrm{Cov}(\delta \v_{1:N}) \left (  
         \p(\theta) \otimes \I_N \right ) \L_N^T \u.
\end{align}
We assume that the velocity estimation errors are i.i.d. over different frames, and in this case we have that $\mathrm{Cov}(\delta \v_{1:N}) = \mathrm{Cov}(\delta\v) \otimes \I_N$, where $\mathrm{Cov}(\delta\v)$ is a $2\times2$ covariance matrix of the velocity estimation error in a single frame. Consequently, 
\begin{align}
    \label{Eq_analysis_az_vel_var_simp}
\mathrm{var}(\delta \theta) &\approx \frac{\u^T \L_N \L_N^T \u}{\|\u\|_2^4}  
\times
\left (  
         \p(\theta)^T \mathrm{Cov}(\delta\v) \p(\theta)
         \right ).
\end{align}


In \eqref{Eq_analysis_az_vel_var_simp} we have arrived at a formula for the SAR angle estimation error as a function of the velocity estimation error covariance. In the next section we derive an analytical expression for the radar's velocity estimation error covariance, $\mathrm{Cov}(\delta\v)$.

\subsection{Velocity Estimation Error Covariance}
\label{sec:analysis2}

In this subsection, we derive a small error analysis for the per frame velocity estimation. Recall that the velocity estimation of the radar is based on minimizing the least squares objective $Q(\tilde{\v}; \tilde{\f},\tilde{\bphi})$, which is define in \eqref{VEL_LS}, with respect to $\tilde{\v}$ \cite{VEL_EST_1}-\cite{VEL_EST_3} (see Section \ref{SYS_MODEL_SEC}).
Note that in order to simplify the notations, in this subsection we omit the frame index $n$. 
Observe that the least squares objective is formulated using $\{\tilde{\phi}_i\}_{i=1}^K$ and $\{\tilde{f}_i\}_{i=1}^K$, which are 
the noisy angle and Doppler measurements of $K$ targets, respectively.
It is assumed that the reflection points angle measurements errors are i.i.d. with variance $\sigma_\phi^2$ and their Doppler frequency measurements are also i.i.d. with variance $\sigma_f^2$.

Let us denote 
the gradient of $Q(\tilde{\v}; \tilde{\f}, \tilde{\bphi})$ (w.r.t. $\tilde{\v}$) by $\q(\tilde{\v}; \tilde{\f}, \tilde{\bphi}) \in \mathbb{R}^{2 \times 1}$. 
The small error analysis utilizes two facts. The first fact is that the estimated velocity, $\hat{\v}$, is a minimizer of $Q(\tilde{\v}; \tilde{\f}, \tilde{\bphi})$ and thus $\q(\hat{\v}; \tilde{\f}, \tilde{\bphi})=\0$. The second fact is that the true velocity, $\v$, is a minimizer of $Q(\v; \f, \bphi)$ and thus $\q(\v; \f, \bphi)=\0$, where $\f$ and $\bphi$ denote the $K \times 1$ vectors of the true Doppler frequencies and direction of arrivals.
We use these properties in the following Taylor series analysis.

The first order Taylor expansion of $\q(\hat{\v}; \tilde{\f}, \tilde{\bphi})$ around $(\v, \f, \bphi)$ is given by
\begin{align}
\label{Eq_small_err_v}
    \q(\hat{\v}; \tilde{\f}, \tilde{\bphi}) &\approx \q(\v; \f, \bphi) + \frac{\partial}{\partial \v}\q(\v; \f, \bphi) \delta\v \\ \nonumber 
    &+ \frac{\partial}{\partial \bphi}\q(\v; \f, \bphi) \delta\bphi + \frac{\partial}{\partial \f}\q(\v; \f, \bphi) \delta\f,
\end{align}
where $\delta\v = \hat{\v} - \v$, $\delta\bphi = \tilde{\bphi} - \bphi$ and $\delta\f = \tilde{\f} - \f$,  
and the approximation is expected to be rather tight in the small error regime (when $\| \delta\v \|_2$, $\| \delta\bphi \|_2$ and $\| \delta\f \|_2$ are small).
Now, observe that from $\q(\hat{\v}; \tilde{\f}, \tilde{\bphi})=\0$ and $\q(\v; \f, \bphi)=\0$ we get
\begin{align}
\label{Eq_small_err_v2}
    -\frac{\partial}{\partial \v}\q(\v; \f, \bphi) \delta\v \approx \frac{\partial}{\partial \bphi}\q(\v; \f, \bphi) \delta\bphi + \frac{\partial}{\partial \f}\q(\v; \f, \bphi) \delta\f.
\end{align}
Denote the derivatives $\frac{\partial}{\partial \v}\q(\v; \f, \bphi)$, $\frac{\partial}{\partial \bphi}\q(\v; \f, \bphi)$ and $\frac{\partial}{\partial \f}\q(\v; \f, \bphi)$ by $\U, \T$ and $\F$, respectively.
Under the reasonable assumptions that the errors $\delta \bphi$ and $\delta\f$ are uncorrelated 
and have zero means we have that
\begin{align}
\label{Eq_small_err_v4}
    \U \mathrm{Cov}(\hat{\v}) \U^T \approx \T \mathrm{Cov}(\tilde{\bphi}) \T^T + \F \mathrm{Cov}(\tilde{\f}) \F^T.
\end{align}
For invertible $\U$ (we elaborate on this below), substituting $\mathrm{Cov}(\tilde{\bphi}) = \sigma_{\phi}^2 \I_K$ and $\mathrm{Cov}(\tilde{\f}) = \sigma_f^2 \I_K$ (the i.i.d. measurements assumption) yields
\begin{align}
\label{Eq_small_err_v5}
     \mathrm{Cov}(\hat{\v}) \approx \sigma_{\phi}^2 \U^{-1} \T  \T^T \U^{-T} + \sigma_f^2 \U^{-1} \F \F^T \U^{-T}.
\end{align}

We turn now to compute the required derivatives:
\begin{align}
\label{Eq_small_err_derivatives}
    &\q(\v; \f, \bphi) = 
    \nabla_{\v} Q(\v; \f, \bphi) = \G(\bphi)^T ( \G(\bphi)\v - \f ), \\ 
    \label{Eq_small_err_derivatives_U}
    &\U = \frac{\partial}{\partial \v}\q(\v; \f, \bphi) = \G(\bphi)^T \G(\bphi), \\ 
    \label{Eq_small_err_derivatives_T}
    &\T[:,i] = \frac{\partial}{\partial \phi_i}\q(\v; \f, \bphi) = \frac{\partial}{\partial \phi_i} \left ( \frac{2}{\lambda} \p(\phi_i)( \frac{2}{\lambda} \p^T(\phi_i)\v - f_i ) \right )  \nonumber \\
    &\hspace{8mm} = \frac{2}{\lambda} \p'(\phi_i)( \frac{2}{\lambda} \p^T(\phi_i)\v - f_i ) + \frac{4}{\lambda^2} \p(\phi_i) \p'^T(\phi_i)\v  \nonumber \\
    &\hspace{8mm} = \frac{4}{\lambda^2} \p(\phi_i) \p'^T(\phi_i)\v,
    \\    
    &\F = \frac{\partial}{\partial \f}\q(\v; \f, \bphi) = - \G(\bphi)^T,   
\end{align}
where we used the definition of $\G(\bphi) \in \mathbb{R}^{K \times 2}$ that appears in \eqref{Eq_G_def}, 
as well as $\p^T(\phi_i) = [\sin{\phi_i}, \cos{\phi_i}]$ and $\p'^T(\phi_i) = [\cos{\phi_i}, -\sin{\phi_i}]$. 
As for the invertibility of $\U$, 
from \eqref{Eq_small_err_derivatives_U} we see that this merely requires having at least two reflections in distinct angles.

Using these derivatives, 
the expression in \eqref{Eq_small_err_v5} for the velocity estimation error covariance can be more concretely expressed as
\begin{align}
\label{Eq_small_err_v5_simp}
     &\mathrm{Cov}(\hat{\v}) \approx \sigma_f^2 \boldsymbol{\Gamma} + \sigma_{\phi}^2 \boldsymbol{\Gamma} \G(\bphi)^T \D^2(\v,\bphi) \G(\bphi) \boldsymbol{\Gamma},
\end{align}
where 
\begin{equation} \label{GAMMA_EQ}
\boldsymbol{\Gamma}=\left ( \G(\bphi)^T \G(\bphi) \right )^{-1},
\end{equation}
and
\begin{align}
\label{Eq_small_err_v5_simp_def}
     \D(\v,\bphi) = \frac{2}{\lambda} \mathrm{diag}\{ \p'^T(\phi_1)\v, \ldots, \p'^T(\phi_K)\v \}
\end{align}
is a $K \times K$ diagonal matrix composed of the tangential component of the velocity w.r.t.~each of the targets.
In \eqref{Eq_small_err_v5_simp} we used $\T = \G^T(\bphi) \D(\v,\bphi)$.

We conclude this subsection with the novel full analytical expression for the SAR angle estimation error, which is presented in \eqref{SAR_ANGLE_ERR_FULL} at the bottom of the next page.
This formula is obtained by substituting \eqref{Eq_small_err_v5_simp} into \eqref{Eq_analysis_az_vel_var_simp}.
\begin{figure*}[b]
\dotfill 
\begin{align} 
\label{SAR_ANGLE_ERR_FULL}
&\mathrm{var}(\delta \theta)\approx
\frac{\u^T \L_N \L_N^T \u}{\|\u\|_2^4} 
\,\p(\theta)^T \Bigl( \sigma_f^2 \boldsymbol{\Gamma}(\bphi) + \sigma_{\phi}^2 \boldsymbol{\Gamma}(\bphi) \G(\bphi)^T \D^2(\v,\bphi) \G(\bphi) \boldsymbol{\Gamma}(\bphi) \Bigr) \p(\theta)
\end{align}
\end{figure*}

\subsection{Insights From the Analysis}\label{ANALYSIS_IMPLICATIONS}

In what follows, we obtain insights from the analytical expression in \eqref{SAR_ANGLE_ERR_FULL}. 
To this end, 
we make a couple of simplifying assumptions: 1) constant velocity along the $N$ frames in the integration interval; and 2) Large number of targets that are uniformly spread in the interval $[-\pi/2,\pi/2]$.
These assumption are used to simplify the first and second terms in \eqref{SAR_ANGLE_ERR_FULL}, respectively.

We begin with
approximating the velocity as constant during the integration interval, i.e., $\v_1=\v_2=...=\v_N=\v$ and equivalently $\V= \1_N \otimes \v^T$. In practice, the velocity 
does not vary significantly over the SAR integration interval.
Thus, assuming that the velocity is fixed and equal to the average velocity does not significantly violate the actual physical problem. 

Recalling the definition of $\u$ in \eqref{Eq_analysis_u_vec}, 
under the constant velocity assumption we have that
\begin{equation}
    \u = \bPhi_N \L_N (\1_N \otimes \v^T) \p'(\theta)= v_t(\theta) \bPhi_N \L_N \1_N,
\end{equation}
where $v_t(\theta) = \v^T \p'(\theta)
         = v_x \cos{{\theta}} - v_y \sin{{\theta}}$. That is, $v_t(\theta)$ is the tangential component of the true velocity w.r.t.~the target at $\theta$.
It follows that the first factor in \eqref{SAR_ANGLE_ERR_FULL} reads as
\begin{align}
\label{CONST_VEL_U_TILDE}
    \frac{\u^T \L_N \L_N^T \u}{\|\u\|_2^4} = \frac{1}{\omega(N) | v_t(\theta) |^2},
\end{align}
where $\omega(N) \triangleq \left ( \frac{\|\L_N^T \bPhi_N \L_N \1_N\|_2^2}{\|\bPhi_N \L_N \1_N\|_2^4} \right )^{-1}$. 
The following lemma shows that $\omega(N) = \Theta(N)$, i.e., there exist two strictly positive constants, $c_1$ and $c_2$, such that $c_1 N \leq \omega(N) \leq c_2 N$ (the highest order of $N$ in $\omega(N)$ is one).


\begin{lemma}
\label{lemma_1_over_N}
We have that $\|\L_N^T \bPhi_N \L_N \1_N\|_2^2 = \Theta(N^5)$ and $\|\bPhi_N \L_N \1_N\|_2^4 = \Theta(N^6)$. Therefore, 
$$
\left ( \frac{\|\L_N^T \bPhi_N \L_N \1_N\|_2^2}{\|\bPhi_N \L_N \1_N\|_2^4} \right )^{-1} = \Theta(N).
$$
\end{lemma}

\begin{proof}
See Appendix~\ref{app_lemma}.
\end{proof}

By substituting \eqref{CONST_VEL_U_TILDE} into \eqref{SAR_ANGLE_ERR_FULL} we obtain that  
\begin{align} \label{SAR_ANGLE_ERR_CONST_VEL}
&\mathrm{var}(\delta \theta)\approx \\ \nonumber
&\hspace{10mm}
\frac{
\p(\theta)^T \Bigl( \sigma_f^2 \boldsymbol{\Gamma} + \sigma_{\phi}^2 \boldsymbol{\Gamma} \G(\bphi)^T \D^2(\v,\bphi) \G(\bphi) \boldsymbol{\Gamma}\Bigr) \p(\theta)
}{\omega(N) | v_t(\theta) |^2}.
\end{align}

We turn now to simplify the second factor in \eqref{SAR_ANGLE_ERR_FULL}
(the numerator of \eqref{SAR_ANGLE_ERR_CONST_VEL}) under the assumption that the number of targets, $K$, is very large and their direction of arrivals $\{ \phi_i \}$ are independently drawn from the uniform distribution $U[-\pi/2,\pi/2]$.

Recalling the definition of $\bGamma$ in \eqref{GAMMA_EQ}, note that
\begin{align}
    &\lim_{K\to\infty} \frac{1}{K}\bGamma^{-1} = \lim_{K\to\infty} \frac{1}{K} \G(\bphi)^T \G(\bphi) \\ \nonumber
    &= \lim_{K\to\infty} \frac{4}{\lambda^2 K} \sum_{i=1}^K \p(\phi_i) \p^T(\phi_i) = \frac{4}{\lambda^2} \mathbb{E}_{\phi} [ \p(\phi) \p^T(\phi) ] \triangleq \frac{4}{\lambda^2} \R_{\p},
\end{align}
where we used the law of large numbers.
Therefore (via arithmetic of limits), the limit of the inverse of $\frac{1}{K}\bGamma^{-1}$ is given by
\begin{align}
    \lim_{K\to\infty} K \bGamma = \frac{\lambda^2}{4} \R_{\p}^{-1} = \frac{\lambda^2}{4} \left ( \mathbb{E}_{\phi} [ \p(\phi) \p^T(\phi) ] \right )^{-1}.
\end{align}
Similarly, recalling that $\D^2(\v,\bphi)$ is a diagonal matrix whose diagonal elements are given by
$|\frac{2}{\lambda} \p'^T(\phi_i)\v|^2 = | \frac{2}{\lambda} v_{t}(\phi_i) |^2$, we have
\begin{align}
  &\lim_{K\to\infty} \frac{1}{K} \G(\bphi)^T \D^2(\v,\bphi) \G(\bphi) \nonumber \\
  &\hspace{10mm} =  \lim_{K\to\infty} \frac{16}{\lambda^4 K} \sum_{i=1}^K v_{t}^2(\phi_i) \p(\phi_i) \p^T(\phi_i)  \nonumber \\
  &\hspace{10mm} = \frac{16}{\lambda^4} \mathbb{E}_{\phi} [ v_{t}^2(\phi) \p(\phi) \p^T(\phi) ] \triangleq \frac{16}{\lambda^4} \R_{|v_t|\p}(\v).
\end{align}
Based on these two limits, for large $K$ we have that \eqref{SAR_ANGLE_ERR_CONST_VEL} is well approximated by
\begin{align}
\label{Eq_asymp_approx}
    \mathrm{var}(\delta \theta) &\approx 
\frac{
\p(\theta)^T \left ( \sigma_f^2 \frac{\lambda^2}{4} \I_2 + \sigma_{\phi}^2 \R_{\p}^{-1} \R_{|v_t|\p}(\v)  \right ) \R_{\p}^{-1} \p(\theta)
}{K \omega(N) | v_t(\theta) |^2}.
\end{align}
Next, to establish a fully analytical expression, let us compute $\R_{\p}$ and $\R_{|v_t|\p}(\v)$:
\begin{align}
\label{Eq_Rg}
\R_{\p} &= \mathbb{E}_{\phi} [ \p(\phi) \p^T(\phi) ] 
= \frac{1}{2} \int_{-\pi/2}^{\pi/2} \frac{d\phi}{\pi}
 \begin{bmatrix}
          2 \sin^2{\phi} & \sin{2\phi} \\
          \sin{2\phi} & 2 \cos^2{\phi}
         \end{bmatrix} \nonumber \\ 
&= 
\frac{1}{2 \pi} 
         \begin{bmatrix}
          \pi & 0 \\
          0 & \pi
         \end{bmatrix}    = 
         \frac{1}{2}\I_2,
\end{align}
\begin{align}
\label{Eq_Rvg}
&\R_{|v_t|\p}(\v) = \mathbb{E}_{\phi} [ v_{t}^2(\phi) \p(\phi) \p^T(\phi) ] \\ \nonumber
&= \frac{1}{2} \int_{-\pi/2}^{\pi/2} \frac{d\phi}{\pi} ( v_x \cos{\phi} - v_y \sin{\phi} )^2
 \begin{bmatrix}
          2 \sin^2{\phi} & \sin{2\phi} \\
          \sin{2\phi} & 2 \cos^2{\phi}
         \end{bmatrix} \\ \nonumber
&= \frac{1}{2} 
         \begin{bmatrix}
          v_x^2/4 + 3 v_y^2/4 & -v_x v_y/2 \\
          -v_x v_y/2 & 3 v_x^2/4 + v_y^2/4
         \end{bmatrix}.           
\end{align}
Substituting \eqref{Eq_Rg} and \eqref{Eq_Rvg} into \eqref{Eq_asymp_approx} we get the following simplified approximate expression  
for the SAR angle error, 
which can be computed for any $\v=[v_x, v_y]^T$,
\begin{align}
    \label{Eq_asymp_approx2}
    &\mathrm{var}(\delta \theta) \approx \nonumber \\ 
    & 
\frac{
\p(\theta)^T \left ( \sigma_f^2 \lambda^2 \I_2 + \sigma_{\phi}^2   
\begin{bmatrix}
          v_x^2 + 3 v_y^2 & -2 v_x v_y \\
          -2 v_x v_y & 3 v_x^2 + v_y^2
         \end{bmatrix} 
\right ) \p(\theta)
}{2 K \omega(N) | v_t(\theta) |^2}.
\end{align} 
Observe that for the typical case where the vehicle moves only along the y-axis (see axis orientation in Fig.\ref{SYS_MODEL}), i.e., $v_x=0$, we further obtain
\begin{align}
\label{Eq_asymp_approx3}
\mathrm{var}(\delta \theta) &\approx 
\frac{
\p(\theta)^T 
         \begin{bmatrix}
          \sigma_f^2\lambda^2 + 3 \sigma_{\phi}^2 v_y^2   & 0 \\
          0 & \sigma_f^2\lambda^2 + \sigma_{\phi}^2 v_y^2 
         \end{bmatrix}
\p(\theta)
}{2 K \omega(N) | v_y \sin{{\theta}} |^2}
 \nonumber \\ 
&= 
\frac{
(\sigma_f^2\lambda^2 + \sigma_{\phi}^2 v_y^2)(\sin^2{\theta}+\cos^2{\theta}) + 2 \sigma_{\phi}^2 v_y^2 \sin^2{\theta}
}{2 K \omega(N) v_y^2 \sin^2{\theta}}
 \nonumber \\
&= 
\frac{
\sigma_f^2\lambda^2 + \sigma_{\phi}^2 v_y^2 (1+ 2 \sin^2{\theta})
}{2 K \omega(N) v_y^2 \sin^2{\theta}}.
\end{align}
From \eqref{Eq_asymp_approx3} it is realized that the SAR angular error is a function of six factors that are detailed next.

The first two factors are the antenna array's angle estimation variance, $\sigma_{\phi}^2$, and the Doppler estimation variance, $\sigma_{f}^2$ (both estimated per short frame). When these two variances increase the velocity estimation error increases and as a result the SAR angle estimation error increases. Note that $\sigma_{\phi}^2$ is inversely proportional to the antenna array aperture size (which is typically small for automotive radars), and $\sigma_{f}^2$ is inversely proportional to the frame duration, $T_f$. The frame duration is limited in order for the approximation of constant target range and speed per frame to hold. 

The third factor is the ego vehicle velocity, $v_y$. 
Note that \eqref{Eq_asymp_approx3} can be expressed as $\sigma_f^2 \frac{c_1}{v_y^2} + \sigma_{\phi}^2 c_2$, where $c_1, c_2>0$ are independent of $\v_y$. In words, the term that is associated with $\sigma_f^2$ reduces as a function of the ego-speed $|v_y|$ while the term, $c_2$, that is associated with $\sigma_\phi^2$ does not depend on $v_y$ (note that $v_y^2$ appears in both numerator and denominator of $c_2$). 
Therefore, increasing $|v_y|$ improves the accuracy of the SAR angle estimation until it reaches a plateau.
The reason for this somewhat surprising phenomenon can be explained by the two counteractive effects of increasing the ego-speed. 
On the one hand, a larger speed increases the synthetic aperture, which improves the accuracy of the SAR angle estimation for a fixed velocity error covariance (recalling the linear link between $\u$ and the velocity, it is realized that the order of the velocity in the denominator of \eqref{Eq_analysis_az_vel_var_simp} is larger than in the numerator). 
On the other hand, an increase in the true speed also amplifies the effect of the array's angular estimation error (whose variance is $\sigma_\phi^2$) on the velocity estimation error. This is realized by increasing the quantity $\D^2(\v,\bphi)$, which increases the velocity estimation covariance in \eqref{Eq_small_err_v5_simp}. 

The fourth factor is the number of targets, $K$. The velocity estimation error decreases when the number of (uniformly drawn) targets increases, and hence, the SAR angle estimation error decrease as well. 

The fifth factor is the number of frames, $N$. Since $\omega(N)$ is proportional to $N$, the SAR angle estimation error reduces with the increase in the number of frames. 
This follows from the fact that increasing $N$ essentially increases the integration interval and hence the synthetic aperture,
which improves the accuracy of the SAR angle estimation for a fixed per-frame velocity error covariance (recall that the velocity errors of the frames are assumed to be statistically independent). 
Note, though, that increasing 
$N$ also increases the latency of the system.

The sixth factor is the target angle. The SAR angle estimation error reduces as the target angle, $\theta$, increases from $0^{\circ}$ to $90^{\circ}$.
This is explained by the fact that \eqref{Eq_asymp_approx3} is derived for a setting where the $\theta$ is measured w.r.t.~the synthetic aperture axis (the y-axis in Fig.~\ref{SYS_MODEL}). 
Like in any antenna array, the sensitivity of the measurements to the target's angle increases in the synthetic array's boresight and decreases in the synthetic array's endfire, and this is aligned with the potential estimation accuracy.


\section{SAR Performance Evaluation}\label{RES_SEC}
In this section we evaluate the automotive SAR angular estimation accuracy by simulation and compare the simulation results with the analytical analysis that was derived in Section \ref{PERFORMANCE_ANALY_SEC}. The SAR was simulated by processing multiple frames as described in Section \ref{SYS_MODEL_SEC}. The individual frame duration was $20msec$, and the number of frames varied in the tests. For simplification, the ego vehicle had constant velocity. Yet, as this is not a requirement of the system, the velocity was estimated independently per frame from the detected reflection points of each frame. In each test scenario there were static reflection points with angles that were randomly spread (with uniform distribution) in the radar field of view of $[-90^{\circ}:90^{\circ}]$. The number of reflection points varied in the experiments. Each detected reflection point had an angle and Doppler estimation error that was Gaussian distributed with zero mean and standard deviation of $\sigma_{\phi}$ and $\sigma_f$, respectively. In all the tests $\sigma_f=50Hz=1/20ms$ corresponding to the frame duration of $20ms$. This is the typical automotive radar frame duration, in which the approximation of constant velocity within the frame holds. The ego vehicle velocity of each frame was estimated by minimizing \eqref{VEL_LS} (least squares estimation), and the SAR coherent integration was obtained from \eqref{SAR_OUT_5} given the velocity estimations. The signal to noise ratio at the matched filter output of each frame, $y_n(\gamma_n,\theta_n,\v_n)$, was $20 dB$.  

Fig.~\ref{SAR_RES_WITHOUT_ERR} shows the SAR angular resolution (the main-lobe width at 3 dB attenuation of the SAR output at \eqref{SAR_OUT_5}) as a function of the reflection point angle with respect to the vehicle motion direction ($\theta$ in Fig.~\ref{SYS_MODEL}) when the velocity is perfectly known (zero velocity error). The figure shows the SAR resolution for three different synthetic apertures, $1m$, $2.5m$ and $3m$. The angular resolution improves as the target angle increases, as well as when the synthetic aperture size increases. State-of-the-art automotive radars achieve a $1^{\circ}$ angular resolution with an array of antenna elements that has $10cm$ physical aperture (without SAR). From Fig.~\ref{SAR_RES_WITHOUT_ERR} it is realized that SAR can significantly improve this resolution when the velocity estimation has zero error and the synthetic aperture is larger than $1m$, or when the aperture has $1m$ but the reflection point angle is larger than $10^{\circ}$.

Next, we analyze the radar's velocity estimation accuracy. 
Fig.~\ref{VEL_ERROR_FIG} shows the velocity error standard deviation as a function of the antenna array angle estimation error standard deviation, $\sigma_{\phi}$, for different cases of ego-velocity and different number of targets. The solid curves display the error of the simulation tests and the dashed curves show the results of the analytical analysis expression in \eqref{Eq_small_err_v5_simp} (specifically, the square root of the trace of the covariance error matrix). It is realized from the figure, that the analytical analysis closely predicts the simulation results. The velocity estimation error increases with the increase in the ego vehicle velocity, and also increases as the number of reflection points reduces. 

We turn next to evaluate the implications of the velocity estimation error on the SAR angular estimation error. Figs.~\ref{SAR_ERROR_VS_ANGLE_VAR}-\ref{SAR_RES_NOF_FRAMES} present the SAR simulation angular root mean square error (RMSE) as a function of the target angle (solid lines referred to as 'Simulation') for different settings. The SAR simulation angular error is the difference between the angle of the SAR output peak (peak of \eqref{SAR_OUT_5}) and the true target angle. The figures also show the analytical analysis RMSE, which is given by the square root of \eqref{SAR_ANGLE_ERR_FULL} (dashed line referred to as 'Analysis'). In Fig.~\ref{SAR_ERROR_VS_ANGLE_VAR} the SAR RMSE is presented for different values of antenna array angular accuracy, $\sigma_{\phi}=1^{\circ},3^{\circ},10^{\circ}$ (corresponding to high-end,  middle-end, and low-end automotive radars). Fig.~\ref{SAR_RES_VS_ABS_VEL} shows the SAR RMSE for different ego vehicle velocities ($\v=3m/s,10m/s,25m/s$). In Fig.~\ref{SAR_RES_NOF_TARGETS} the SAR RMSE for different number of targets is shown ($K=2,5,10$). Finally, in Fig.~\ref{SAR_RES_NOF_FRAMES} the SAR performance for different number of frames is presented ($N=2,5,10$). 

In all the Figs.~\ref{SAR_ERROR_VS_ANGLE_VAR}-\ref{SAR_RES_NOF_FRAMES} the analytical analysis closely predicts the simulation results. In all figures it is observed that the SAR angular error reduces with the target angle. From Fig.~\ref{SAR_ERROR_VS_ANGLE_VAR} it is observed that the SAR RMSE increases with $\sigma_{\phi}$. From Fig.~\ref{SAR_RES_VS_ABS_VEL} it is realized that the SAR RMSE reduces with the increase in the ego vehicle velocity. The RMSE ratio between $v=10m/s$ and $v=3m/s$ is about $2$ and the RMSE ratio between $v=25m/s$ and $v=10m/s$ is about $1.2$, thus the RMSE reduction factor becomes smaller as the velocity increases. Figs.~\ref{SAR_RES_NOF_TARGETS}-\ref{SAR_RES_NOF_FRAMES} show that the SAR angle estimation RMSE reduces with the number of targets and the number of frames. All these phenomena agree and confirm the insights that were obtained from the analytical analysis in Section \ref{ANALYSIS_IMPLICATIONS}.

Next, we analyze the performance gain in the angle estimation of the SAR with respect to angle estimation with the antenna array. Fig.~\ref{RATIO_SIG_THETA_AND_SAR_RMSE} presents the ratio between the array antenna angle estimation RMSE, $\sigma_{\phi}$, and the analytical analysis SAR RMSE (square root of \eqref{SAR_ANGLE_ERR_FULL}), for the different system configurations that were presented in Figs.~\ref{SAR_ERROR_VS_ANGLE_VAR}-\ref{SAR_RES_NOF_FRAMES}. When this ratio is larger than one the SAR performance is better than the antenna array, and when this ratio is smaller than one then the SAR performs worse than the antenna array. In all the cases in the figure the antenna array angular RMSE was $\sigma_{\phi}=1^{\circ}$, which is the angular resolution of state-of-the-art automotive radars. It is realized from the figure that the conditions in which SAR attains a performance gain over state-of-the-art automotive radars are limited. When the ego-velocity is small ($3m/s$, green plot), or the number of targets is small ($K=2$, red plot), or the number of frames is small ($N=2$, brown plot) then the SAR does not obtain a significant performance improvement over the antenna array. Moreover, in these conditions the SAR obtains worse performance than the antenna array when the target angle is smaller than $10^{\circ}$.
For velocity of $10m/s$, 5 targets and 5 frames (black curve) the SAR attains a significant performance gain over the antenna array only for target angles above $15^{\circ}$, and attains a maximal gain factor of 3 for target angle above $40^{\circ}$. For higher velocities or large number of targets or larger number of frames the SAR attains a larger performance gain (blue, cyan and magenta plots).

Fig.~\ref{RATIO_SAR_RES_AND_SAR_RMSE} presents the SAR performance degradation due to the inaccurate velocity estimation, which is expressed by the ratio between the SAR RMSE and the SAR angular resolution when the velocity is known (the angular resolution as was calculated for Fig.~\ref{SAR_RES_WITHOUT_ERR}). The figure presents this ratio for the test cases in Fig.~\ref{RATIO_SIG_THETA_AND_SAR_RMSE}. When the ratio is significantly smaller than one then the velocity estimation error does not degrade the SAR performance significantly and the SAR reaches the potential performance of the extended aperture. It is observed from the figure that in most cases and in most target angles the SAR does not attain the extend aperture performance. For target angles below $20^{\circ}$ the SAR RMSE is significantly larger than the potential SAR resolution. For target angles above $50^{\circ}$ the SAR RMSE is smaller than the SAR resolution in most test cases.

Figs.~\ref{RATIO_SIG_THETA_AND_SAR_RMSE}-\ref{RATIO_SAR_RES_AND_SAR_RMSE} show that from all the tested configurations, the one with maximal number of targets (magenta plot with 15 targets) achieves maximal gain with respect to the antenna array and also attains the smallest performance gap from the extended aperture performance with known velocity. 
Increasing the SAR integration duration ($N$) also provides a relatively large performance gain with respect to the antenna array (cyan plot in Fig.~\ref{RATIO_SIG_THETA_AND_SAR_RMSE}), however this comes with the price of increasing the radar detection delay, which is a major issue in collision risk situation. We deduce from the figures in this section that applying SAR in automotive scenarios is mostly beneficial in the following conditions: 
\begin{enumerate}
\item Crowded urban environments with a large number of static reflection points.
\item Moderate to high ego vehicle speed ($\|\v\|>10m/s$).
\item Target angles with offset from the driving direction ($\theta>10^{\circ}$).
\end{enumerate}
When these conditions are not met the SAR performance benefit for automotive applications is insignificant and possibly is not worth the increase in processing complexity and increase in the detection delay that comes with its implementation.

\begin{figure}
  \centering
  \begin{subfigure}[b]{1\linewidth}
    \centering\includegraphics[width=260pt]{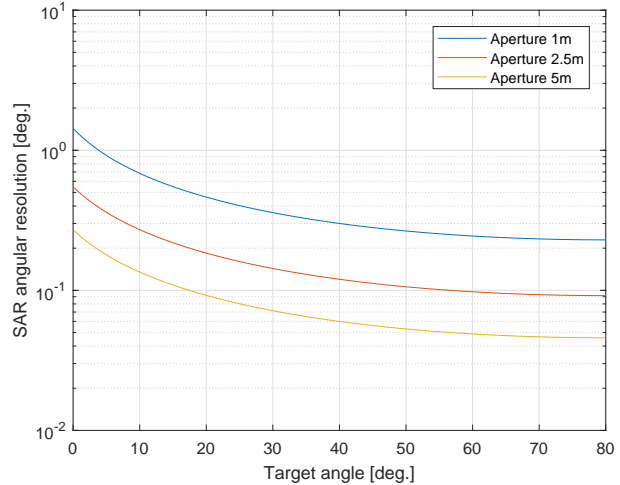}    
  \end{subfigure}
  \caption{SAR angular resolution (the $3dB$ beam-width) when velocity estimation has zero error.}
\label{SAR_RES_WITHOUT_ERR}
\end{figure}
\begin{figure}
  \centering
  \begin{subfigure}[b]{1\linewidth}
    \centering\includegraphics[width=260pt]{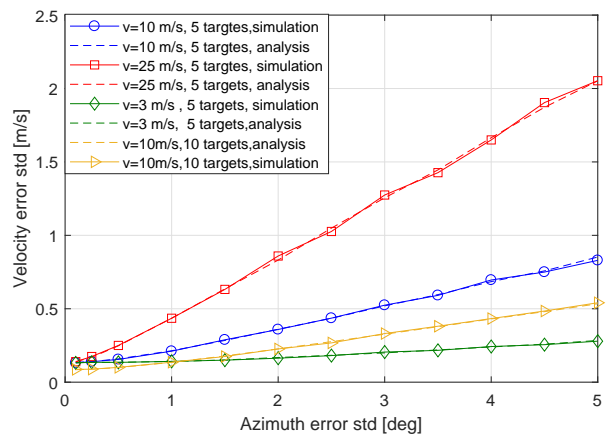}    
  \end{subfigure}
  \caption{Velocity estimation error.}
\label{VEL_ERROR_FIG}
\end{figure}

\begin{figure}
  \centering
  \begin{subfigure}[b]{1\linewidth}
    \centering\includegraphics[width=260pt]{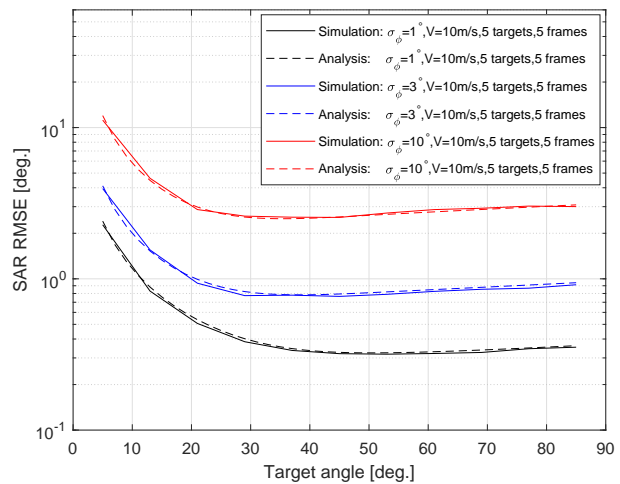}
  \end{subfigure}
  \caption{SAR angle estimation RMSE for different physical antenna array angle estimation RMSE ($\sigma_{\phi}$)}
\label{SAR_ERROR_VS_ANGLE_VAR}
\end{figure}

\begin{figure}
  \centering
  \begin{subfigure}[b]{1\linewidth}
    \centering\includegraphics[width=260pt]{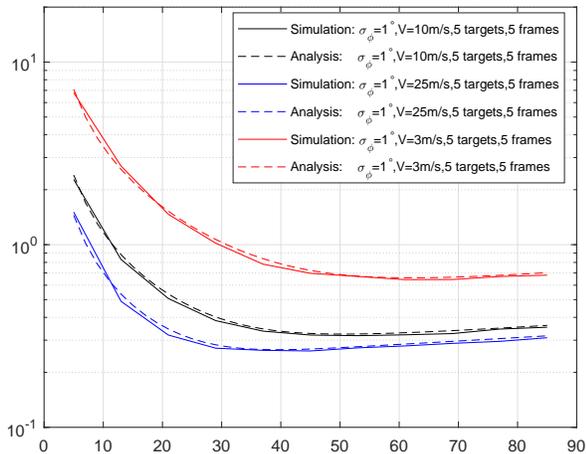}    
  \end{subfigure}
  \caption{SAR angle estimation RMSE for different ego vehicle velocities.}
\label{SAR_RES_VS_ABS_VEL}
\end{figure}

\begin{figure}
  \centering
  \begin{subfigure}[b]{1\linewidth}
    \centering\includegraphics[width=260pt]{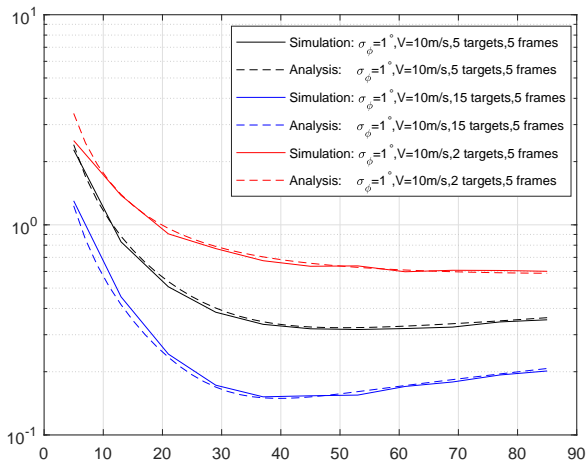}
  \end{subfigure}
  \caption{SAR angle estimation RMSE for different number of targets (number of reflection points)}
\label{SAR_RES_NOF_TARGETS}
\end{figure}

\begin{figure}
  \centering
  \begin{subfigure}[b]{1\linewidth}
    \centering\includegraphics[width=260pt]{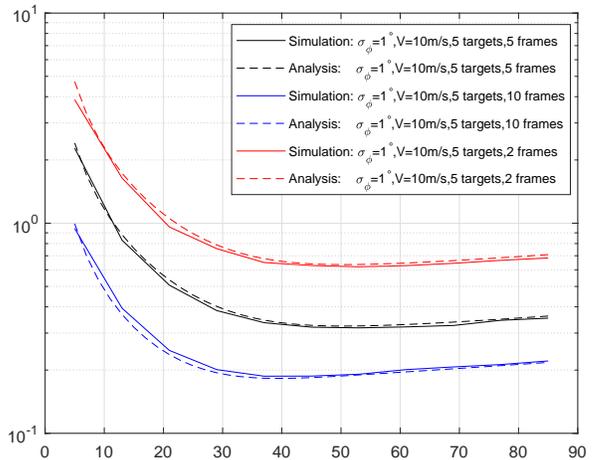} 
  \end{subfigure}
  \caption{SAR angle estimation RMSE for different number of frames (different integration duration)}
\label{SAR_RES_NOF_FRAMES}
\end{figure}

\begin{figure}
  \centering
  \begin{subfigure}[b]{1\linewidth}
    \centering\includegraphics[width=260pt]{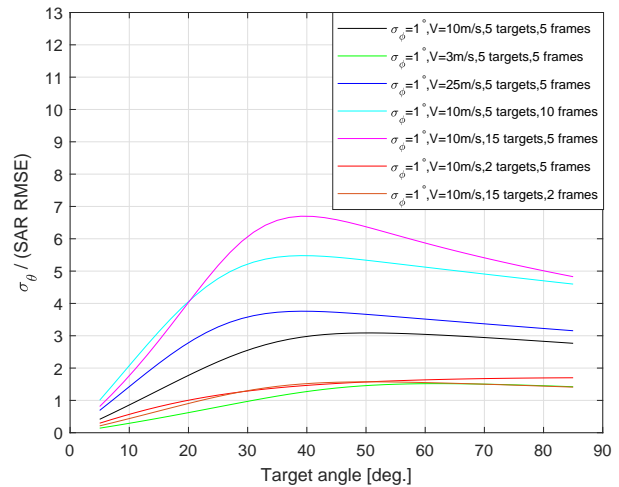}  
  \end{subfigure}
  \caption{SAR performance gain with respect to the physical antenna array. Ratio between the physical antenna array angle estimation RMSE ($\sigma_{\phi}$) and the SAR angle estimation RMSE (given by the square root of \eqref{SAR_ANGLE_ERR_FULL})}
\label{RATIO_SIG_THETA_AND_SAR_RMSE}
\end{figure}

\begin{figure}
  \centering
  \begin{subfigure}[b]{1\linewidth}
    \centering\includegraphics[width=260pt]{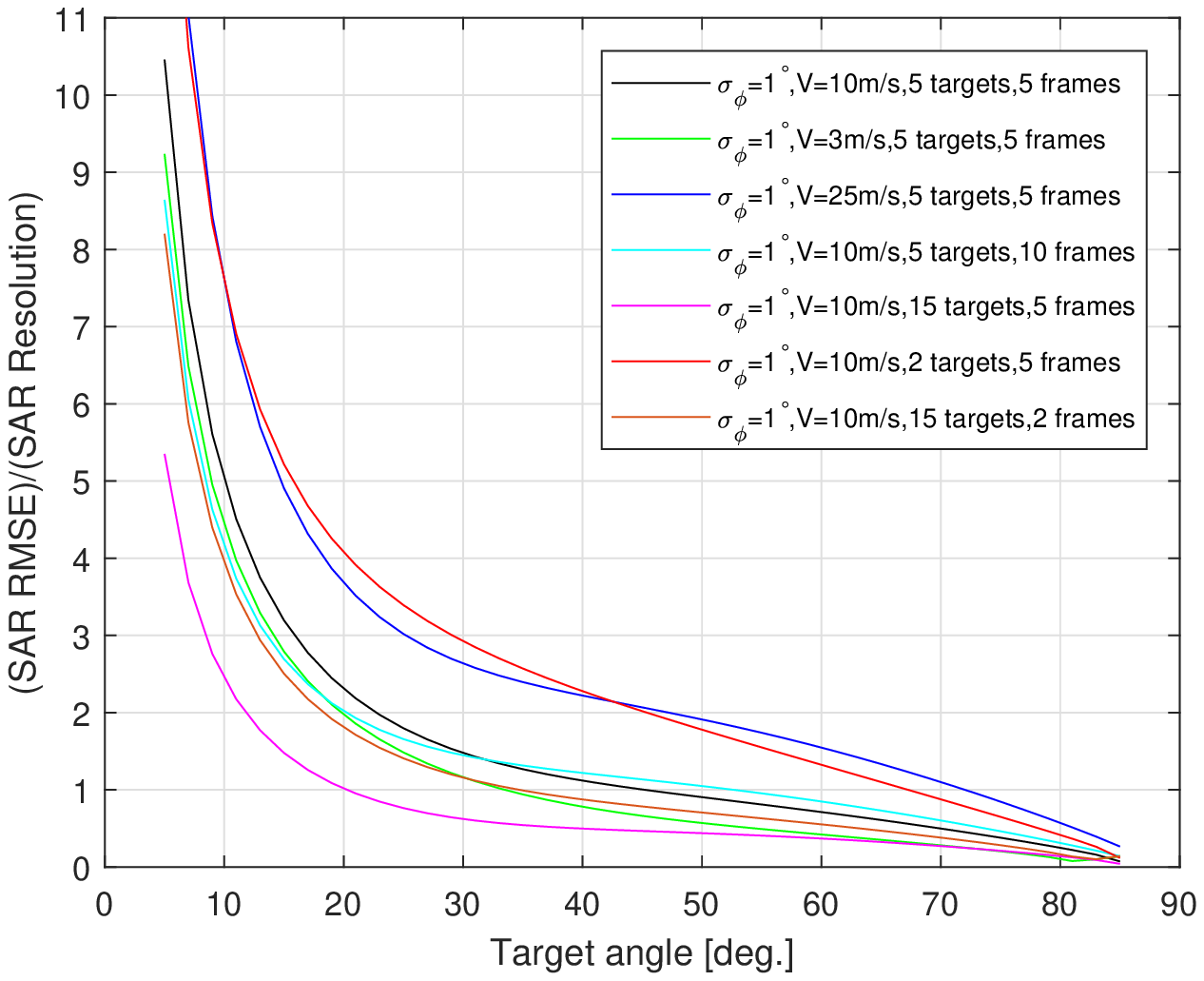}  
  \end{subfigure}
  \caption{SAR performance degradation due to velocity estimation error. Ratio between SAR angle estimation RMSE (square root of \eqref{SAR_ANGLE_ERR_FULL}) and the SAR angular resolution (the $3dB$ beam-width) when the velocity is known.}
\label{RATIO_SAR_RES_AND_SAR_RMSE}
\end{figure}

\section{Conclusion}\label{CONC_SEC}
In this paper, we analyzed the angle estimation error of automotive SAR, which is due to the error in the velocity estimation, where the velocity estimation is obtained from the detections of the radar's physical antenna array. An analytical expression for the angle estimation error of SAR was derived, which closely predicts the simulation performance, and provides insights on the limitations of SAR for automotive applications. It was realized from the analytical analysis that SAR attains a significant performance gain when there is a large number of reflection points (with the same velocity), the ego vehicle speed is moderate to high, and the target angle has at least $10^{\circ}$ offset from the driving direction. 
These conditions are apparent when driving in urban environments that are crowded with static objects. In other driving scenarios such as highway with small number of reflections, or parking at low speed, these conditions may not hold. In this case, the SAR performance gain is small and may not be worth the increase in complexity and detection delay that comes from SAR. Moreover, when there are only few targets, and the driving speed is very low the SAR performs even worse than the physical antenna array.

The limitations of SAR for the automotive applications that were highlighted in this paper may possibly be overcome by utilizing auxiliary sensor such as very accurate GNSS and IMU sensors. In future work we suggest to leverage the work of this paper for analyzing the automotive SAR performance for systems that estimate the velocity with a hybrid of radar and auxiliary sensors.

\appendices

\section{Proof of Lemma~\ref{lemma_1_over_N}}
\label{app_lemma}

To show that 
$\|\L_N^T \bPhi_N \L_N \1_N\|_2^2 = \Theta(N^5)$ (i.e., that the highest order of $N$ in this expression is 5) and $\|\bPhi_N \L_N \1_N\|_2^4 = \Theta(N^6)$,
we express each of these terms as a polynomial of $N$.
We do it here for odd values of $N$, but it is straightforward to establish this result in a similar way also for even values of $N$. 

We start with $\|\bPhi_N \L_N \1_N\|_2^4$. 
Recall that $\L_N \1_N$ is a vector of entries $1,...,N$ and that applying $\bPhi_N$ on $\L_N \1_N$ subtracts the mean $\frac{N+1}{2}$ from all its entries. Thus, we have that
\begin{align}
    \|\bPhi_N \L_N \1_N\|_2^4 &= \left ( \sum_{i=1}^N \left ( i - \frac{N+1}{2} \right )^2 \right )^2 \\ \nonumber
    &= \left ( 2 \sum_{i=1}^{\frac{N-1}{2}} i^2 \right )^2  \\ \nonumber
    &= 4 \left ( \frac{N\frac{N-1}{2}\frac{N+1}{2}}{6} \right )^2   \\ \nonumber
    &= \frac{N^2(N^2-1)^2}{144} = \Theta(N^6),
\end{align}
where the third equality follows from the well-known result on the sum of $1,2^2,...,n^2$.

We turn to handle $\|\L_N^T \bPhi_N \L_N \1_N\|_2^2$.
First, note that the $i^{th}$ element in $\L_N^T \bPhi_N \L_N \1_N$ is 
\begin{align}
    [\L_N^T \bPhi_N \L_N \1_N]_i &= \sum_{k=i}^N \left ( k - \frac{N+1}{2} \right ) \\ \nonumber
    &= \frac{(N+i)(N-i+1)}{2} + (N-i+1)\frac{N+1}{2} \\ \nonumber
    &= \frac{(N-i+1)(i-1)}{2},
\end{align}
where in the second equality we used the well-known result on the sum of $i,i+1,...,N$.
Equipped with the result, we have
\begin{align}
\label{Eq_app_lemma}
    &\|\L_N^T \bPhi_N \L_N \1_N\|_2^2 = \frac{1}{4} \sum_{i=1}^N \left ( (N+2)i - i^2 - (N+1) \right )^2= \\ \nonumber
    & \frac{1}{4} \sum_{i=1}^N \Big [ (N+2)^2 i^2 + i^4 + (N+1)^2 -2(N+1) i^3 \\ \nonumber
    &\hspace{5mm} - 2(N+1)(N+2)i + 2(N+1)i^2 \Big ]= \\ \nonumber
    & \frac{1}{4} \Big [ \big ((N+2)^2+2(N+1) \big ) \sum_{i=1}^N i^2 + \sum_{i=1}^N i^4 + N(N+1)^2  \\ \nonumber
    &\hspace{5mm} -2(N+1) \sum_{i=1}^N i^3 - 2(N+1)(N+2) \sum_{i=1}^N i \Big ].
\end{align}
Combining this expression with the following well-known results of series 
\begin{align*}
    \sum_{i=1}^N i &= \frac{N(N+1)}{2} \\ \nonumber
    \sum_{i=1}^N i^2 &= \frac{N(N+1)(2N+1)}{6} \\ \nonumber
    \sum_{i=1}^N i^3 &= \frac{N^2(N+1)^2}{4} \\ \nonumber
    \sum_{i=1}^N i^4 &= \frac{N(N+1)(2N+1)(3N^2+3N-1)}{30} \\ \nonumber    
\end{align*}
it can be seen that $N^5$ is the highest order of $N$ that appears in \eqref{Eq_app_lemma}, and the coefficient of $N^5$ in \eqref{Eq_app_lemma} is given by 
$$
\frac{1}{4} \left [ \frac{2}{6} + \frac{6}{30} - \frac{2}{4} \right ] = \frac{1}{120}.
$$
This implies that $\|\L_N^T \bPhi_N \L_N \1_N\|_2^2 = \Theta(N^5)$.

\end{document}